\documentclass[twoside,english]{elsarticle}
\usepackage[T1]{fontenc}
\usepackage[latin9]{inputenc}
\usepackage{amsthm}
\usepackage{graphicx}

\makeatletter

\providecommand{\tabularnewline}{\\}

\theoremstyle{plain}
\newtheorem{thm}{\protect\theoremname}
\ifx\proof\undefined
\newenvironment{proof}[1][\protect\proofname]{\par
\normalfont\topsep6\p@\@plus6\p@\relax
\trivlist
\itemindent\parindent
\item[\hskip\labelsep\scshape #1]\ignorespaces
}{%
\endtrivlist\@endpefalse
}
\providecommand{\proofname}{Proof}
\fi

\journal{arXiv.org (preprint)}

\makeatother

\usepackage{babel}
\providecommand{\theoremname}{Theorem}

\begin{document}

\begin{frontmatter}{}

\title{On the optimality of ternary arithmetic for compactness and hardware
design\tnoteref{t1}}

\tnotetext[t1]{Some explanations and derivations in equations are presented in explicit
detail, in order to be more readable as lecture material for undergraduate
students.}

\author[rvt]{Harris~V.~Georgiou (MSc,~PhD)\fnref{fn1}}

\ead[url]{http://xgeorgio.info}

\fntext[fn1]{H.G. is a post-doc researcher with NKUA and an R\&D associate with
University of Piraeus, Greece. His main area of expertise is Machine
Learning \& A.I., Signal Processing and Medical Imaging. He is also
working as an independent R\&D consultant since 1998. }

\address[hvg]{Department of Informatics \& Telecommunications (DIT),\\
National Kapodistrian University of Athens (NKUA/UoA), Greece}
\begin{abstract}
In this paper, the optimality of ternary arithmetic is investigated
under strict mathematical formulation. The arithmetic systems are
presented in generic form, as the means to encode numeric values,
and the choice of radix is asserted as the main parameter to assess
the efficiency of the representation, in terms of information compactness
and estimated implementation cost in hardware. Using proper formulations
for the optimization task, the universal constant $e$ (base of natural
logarithms) is proven as the most efficient radix and ternary is asserted
as the closest integer choice.\end{abstract}

\end{frontmatter}{}

\section{Introduction}

The term \emph{arithmetic system} refers to the way a number is represented
as a sequence of symbols associated with a specific power series.
More specifically, a \emph{radix} $r$ is selected as the the base
of the arithmetic system and every number is expressed as a sum of
powers of this radix.

As humans, we learn to count in the \emph{decimal} arithmetic system,
i.e., using powers of 10, purely for practical reasons. Children learn
basic arithmetic by using their ten fingers, thus each time they count
to 10, a carrier is created and added to the next power index. For
example, the number 1,234 is actually a shortcut to the full representation:
\[
1234_{10}=1\cdot1.000+2\cdot100+3\cdot10+4=1\cdot10^{3}+2\cdot10^{2}+3\cdot10^{1}+4\cdot10^{0}
\]
The proper representation of the number orders the coefficients for
each radix power in left-to-right ranking and includes a subscript
displaying the radix. The representation may include a sign, a fractional
point and negative powers, in order to represent any real number:
\[
-12.34{}_{10}=\left(-1\right)\left(1\cdot10^{1}+2\cdot10^{0}+3\cdot10^{-1}+4\cdot10^{-2}\right)
\]

Any radix may be used instead of 10 to represent the same number in
different arithmetic systems. The general formula for representing
$x\epsilon\mathbf{R}$ is:
\begin{equation}
x=...+c_{2}\cdot r{}^{2}+c_{1}\cdot r{}^{1}+c_{0}\cdot r{}^{0}+c_{-1}\cdot r{}^{-1}+...=\sum_{k=-m}^{n}c_{k}\cdot r{}^{k}\label{eq:arithm-systems}
\end{equation}

where $r$ is the radix used, $k$ are the power indices and $c_{k}$
are appropriate coefficients to represent a given number $x$ (sign
is omitted here). It is clear that for a specified positive radix,
the corresponding representation for \emph{x} is unique. Theoretically,
both the coefficients and the radix use can be negative and/or non-integers,
although this may complicate the representation. On the other hand,
choosing the proper radix can result in more efficient and \emph{compact}
representation. For example, $1024_{10}=1\cdot2^{10}=10000000000_{2}$
but also $1024_{10}=1\cdot1024^{1}=10{}_{1024}$ and all three representations
can be used equivalently.

In every arithmetic system, the radix $r$ determines the valid range
for the $c_{k}$ coefficients, since these are upper-bounded in every
power position by the next one, and for each valid value a unique
symbol is required. Hence, in the decimal system the valid symbols
are $\left\{ 0,...,9\right\} ,$ in the binary system these are $\left\{ 0,1\right\} $,
in the hexadecimal these are $\left\{ 0,...,9,A,...,F\right\} $,
etc. In general, if the radix is chosen as a positive integer $r>1$
and $c_{k}\geq0$ then the valid range is $0\leq c_{k}\leq r-1$,
since when $c_{k}=r$ then it results in carrier $c_{k+1}=1$ and
$c_{k}=0$. 

From these basic common properties of all arithmetic systems it is
clear that using a large radix results in more compact representation,
but a larger set of symbols is required. In contrast, using a small
radix results in longer representation, but a smaller set of symbols.
This can be easily verified with arithmetic systems where the radix
of one is a multiplication of the other, for example $255_{10}=11111111_{2}=FF_{16}$.

\section{Problem definition}

The problem of defining an optimal arithmetic representation has been
around for many decades, especially in the very early years of computing
technologies. In particular, mathematicians and engineers have been
trying to propose the best arithmetic system that would be used as
the basis for computers back in the '50s and '60s \cite{Hayes2001}.

The earliest published discussion on this subject is probably in the
1950 book by U.S. Navy and Engineering Research Associates \cite{ERA1950}.
In particular, the product of the radix and the size of the symbols
set was considered a good predictor of the cost of the hardware (electronic
components) required to build digital computers. In particular, the
best radix for optimal arithmetic and compact number representation
was calculated close to 3, thus producing the notion of a \emph{ternary}
system.

Before this assertion is formally proven, one more comment should
be noted with regard to integer power series and their summations.
In Eq.\ref{eq:arithm-systems} the coefficients $c_{k}$ can take
any value in the valid range $0\leq c_{k}\leq r-1$. Without loss
of generality, we limit the problem to positive integers, hence the
maximum number that can be represented with a fixed width $w+1$ is
when $c_{k}=r-1\ ,\ \forall k=\left\{ 0,...,w\right\} $. The +1 is
for the zero-index power $r^{0}$. Hence, we can now calculate the
maximum number that can be represented with $w+1$ symbols (size)
in an arithmetic system of radix $r$ as follows:
\begin{thm}
\label{thm:arithm-sum}The maximum number $U$ that can be represented
with $w+1$ symbols (size) in an arithmetic system of radix $r>1$
is upper-bounded by $r^{w+1}.$\end{thm}
\begin{proof}
Using Eq.\ref{eq:arithm-systems} and setting $c_{k}=r-1\ ,\ \forall k=\left\{ 0,...,w\right\} $,
it results to:

\begin{equation}
U=\sum_{k=0}^{w}\left(r-1\right)r{}^{k}=\sum_{k=0}^{w}r{}^{k+1}-\sum_{k=0}^{w}r{}^{k}=r^{w+1}+\left(1-1\right)\sum_{k=1}^{w}r{}^{k+1}+r^{0}=r^{w+1}-1\label{eq:arithm-limit}
\end{equation}

\end{proof}
In should be noted that the result in Eq.\ref{eq:arithm-limit} is
compatible with the sum of a geometric series \cite{Spiegel2002}
via:

\begin{equation}
\sum_{k=1}^{n}a{}^{k}=\frac{a^{n+1}-1}{a-1}\label{eq:arithm-series-sum}
\end{equation}

and thus:
\begin{eqnarray*}
\sum_{k=0}^{w}\left(r-1\right)r{}^{k} & = & \sum_{k=0}^{w}r{}^{k+1}-\sum_{k=0}^{w}r{}^{k}\\
 & = & \frac{r^{w+2}-1}{r-1}-1-\frac{r^{w+1}-1}{r-1}\\
 & = & \frac{r^{w+1}\left(r-1\right)}{r-1}-\frac{r-1}{r-1}=r^{w+1}-1
\end{eqnarray*}

According to Eq.\ref{eq:arithm-limit}, the maximum number that can
be represented with $w+1$ symbols (set size) in an arithmetic system
of radix $r$ is upper-bounded by $r^{w+1}$. This result from Theorem
\ref{thm:arithm-sum} is very useful indeed, since it is valid for
any radix $r$. For example, the maximum four-digit decimal number
is ($w=3$): $9999_{10}=10000_{10}-1=10^{4}-1$.

\section{Formalizations and proofs}

In the previous section the problem was clearly defined: \emph{Given
that the radix and the size of number representation in any arithmetic
system are inversely associated, how can we choose the ``best''
combination?} 

To understand the nature of this problem and the way to solve it,
a number can be viewed as a collection of distinct items, a set with
size that corresponds to this value. This collection can then be organized
in the form of a tree, with nodes of equal size and depth appropriate
to completely accommodate this set. Such an example for the number
39 is presented in Figure \ref{fig:ternary-tree}.

\begin{center}
\begin{figure}
\centering{}\caption{\label{fig:ternary-tree}A sample packed tree organization for exactly
39 items.}
\includegraphics[scale=0.17]{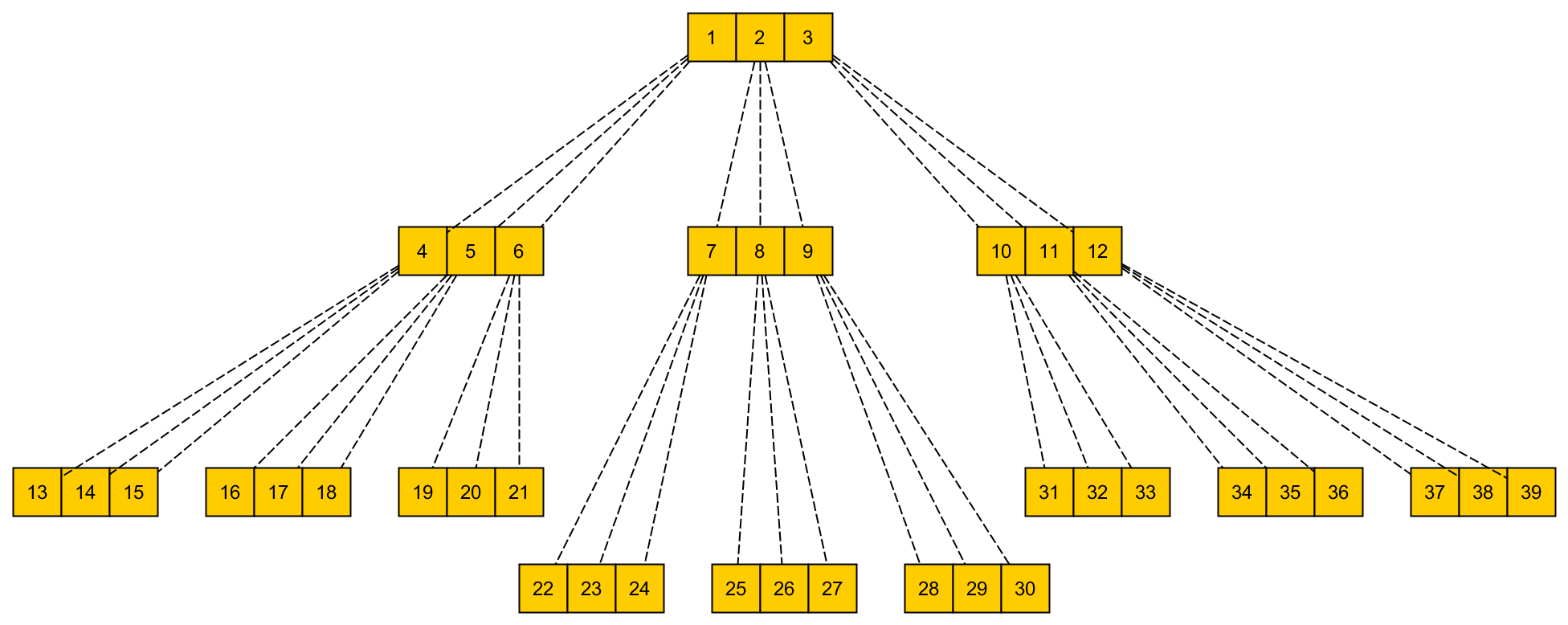}
\end{figure}

\par\end{center}

In Figure \ref{fig:ternary-tree} each node contains exactly three
items and each one of them is the root for three similar sub-trees.
The tree presented is filled completely and no more items can be added
without expanding it to greater depth. In other words, there are 39
items filling the entire tree up to a depth of three, thus the tree
is \emph{packed} and no representation more compact than this can
be created, considering this specific node size and sub-trees structure.

It is easy to see that depth $d$ contains $m^{d}$ items: $m^{1}=3$,
$m^{2}=3\cdot3=9$ and $m^{3}=3\cdot3\cdot3=27$. Hence, the sum of
all the items in the tree is: 3+9+27=39. Not surprisingly, this is
exactly what Eq.\ref{eq:arithm-series-sum} gives when $a=3$ and
$n=3$:
\[
\sum_{k=1}^{3}3{}^{k}=\frac{3^{3+1}-1}{3-1}=\frac{81-1}{2}=40
\]

where the sum 40 includes the zero-power item or ``root'' of the
tree of Figure \ref{fig:ternary-tree} (not shown). In other words,
if the upper bound that is asserted by Theorem \ref{thm:arithm-sum}
is considered as the size of a collection of distinct items, then
radix $r$ and width $w$ are associated with the node size $m$ and
the required depth $d$, respectively, if a packed tree organization
is employed. In this case, the tree depth $d$, therefore the width
$w$ too, can be estimated by employing Eq.\ref{eq:arithm-series-sum}
and using the base-$m$ logarithm:
\[
40=\frac{3^{d+1}-1}{3-1}\Rightarrow3^{d+1}=40\cdot2+1\Rightarrow d+1=\log_{3}81=4
\]

or in a more general form, using Eq.\ref{eq:arithm-series-sum}:
\begin{equation}
a^{n+1}=1+\left(a-1\right)\sum_{k=1}^{n}a{}^{k}\Rightarrow d=\log_{a}\left(1+\left(a-1\right)\sum_{k=1}^{n}a{}^{k}\right)-1\label{eq:series-sum-inv}
\end{equation}

What Eq.\ref{eq:series-sum-inv} tells us is that the (minimum) depth
of such a packed tree is directly related to the node size. Considering
this tree organization as a metaphor for the upper bound $U=r^{w+1}-1$
that is asserted by Theorem \ref{thm:arithm-sum}, the width $w$
can be calculated in a similar way as in Eq.\ref{eq:series-sum-inv}:
\begin{equation}
U=r^{w+1}-1\Rightarrow\left(w+1\right)\log_{r}r=\log_{r}\left(U+1\right)\Rightarrow w=\log_{r}\left(U+1\right)-1\label{eq:arithm-sum-inv}
\end{equation}

For example, the number $U=255_{10}$ represented with $r=2$ (binary)
requires: $w=\log_{2}\left(256\right)-1=7$. This is exactly what
Theorem \ref{thm:arithm-sum} states, with $w+1=8$ coefficients for
the power indices $\left\{ 2^{7},...,2^{0}\right\} $ and $U=11111111_{2}=255_{10}$. 

Using the tree structure metaphor and the problem formalization presented
above as an optimization task, the following theorem describes the
solution for the ``best'' arithmetic system in terms of information
packing and representation efficiency:
\begin{thm}
\label{thm:E1-base-n}Under the minimization criterion $E_{1}\left(r,w\right)=r\cdot w$
and subject to $r^{w}=C\neq0$ (constant), where $r>1$ is the radix
and $w\ \epsilon\left[0,\ r-1\right]$, the optimal information packing
for number representation is achieved for an arithmetic system with:
$r=e$, the base of the natural logarithm ($e=2.718,281,828,459...$).\end{thm}
\begin{proof}
The goal is to minimize: $E_{1}\left(r,w\right)=r\cdot w$, s.t. $r^{w}=C\neq0$
(constant). Since $r>1$, from the constraint we have: 

\begin{equation}
r^{w}=C\Leftrightarrow w\cdot\log_{a}r=\log_{a}C\Leftrightarrow w=\frac{\log_{a}C}{\log_{a}r}\label{eq:w-param-log}
\end{equation}

Therefore, substituting Eq.\ref{eq:w-param-log} in the minimization
criterion, then taking the first derivative \cite{Spiegel2002} and
calculating for root(s), gives:

\begin{eqnarray}
\frac{\partial\left(E_{1}\left(r,w\right)\right)}{\partial r}=0 & \Rightarrow & \frac{\partial}{\partial r}\left(r\cdot\frac{\log_{a}C}{\log_{a}r}\right)=\log_{a}C\left(\frac{\frac{\partial r}{\partial r}\cdot\log_{a}r-r\cdot\frac{\partial\log_{a}r}{\partial r}}{\log_{a}^{2}r}\right)=0\nonumber \\
 & \Rightarrow & \frac{\ln C}{\ln a}\left(\frac{\frac{\ln r}{\ln a}-\frac{r}{r\cdot\ln a}}{\left(\frac{\ln r}{\ln a}\right)^{2}}\right)=\frac{\ln C}{\ln a}\left(\frac{\frac{1}{\ln a}\left(\ln r-1\right)}{\frac{1}{\ln a}\cdot\frac{\ln^{2}r}{\ln a}}\right)=0\label{eq:E1-analytical0}\\
 & \Rightarrow & \ln C\left(\frac{1}{\ln r}-\frac{1}{\ln^{2}r}\right)=\frac{\ln C}{\ln r}\left(1-\frac{1}{\ln r}\right)=0\label{eq:E1-analytical1}
\end{eqnarray}

From Eq.\ref{eq:E1-analytical0} it becomes clear that the calculation
will eventually reduce to the simple: $\ln r-1=0$. It is worth noting
that in Eq.\ref{eq:E1-analytical1} there is no reference to $a$,
i.e., the base of the logarithm used in Eq.\ref{eq:w-param-log} is
indeed irrelevant, as expected. From here, the derivation of the root
is straight-forward:

\begin{eqnarray}
\frac{\partial\left(E_{1}\left(r,w\right)\right)}{\partial r}=0 & \Rightarrow & \frac{\ln C}{\ln r}\left(1-\frac{1}{\ln r}\right)=0\nonumber \\
 & \Rightarrow & \ln r=1\Rightarrow e^{\ln r}=e^{1}\Rightarrow r=e\label{eq:radix-e-final}
\end{eqnarray}

\end{proof}
Figure \ref{fig:cost-function-plot} illustrates the plot of $E_{1}\left(r,w\right)$
cost function (using $\ln\left(.\right)$ and $\ln C=1$) against
radix $r$ and the optimal solution at $r=e$. This is the representation
``cost'' measured as the product of node size and tree depth, i.e.,
analogous to the 2-D ``size'' of the associated tree, if a specific
number (upper bound) $C$ is treated as a collection of distinct items.
The shape of the function makes it clear that there is a well-defined
and unique minimum, at the point where the radix is equal to the base
of the natural logarithm. It is very important to point out that,
under this very generic formalization of the problem, the optimal
solution is independent to both the actual upper bound $C$, as well
as the base of the logarithm used in the proof.

\begin{center}
\begin{figure}
\centering{}\caption{\label{fig:cost-function-plot}Plot of the $E_{1}\left(r,w\right)$
cost function ($\ln C=1$) and the optimal solution at $r=e$.}
\includegraphics[scale=0.75]{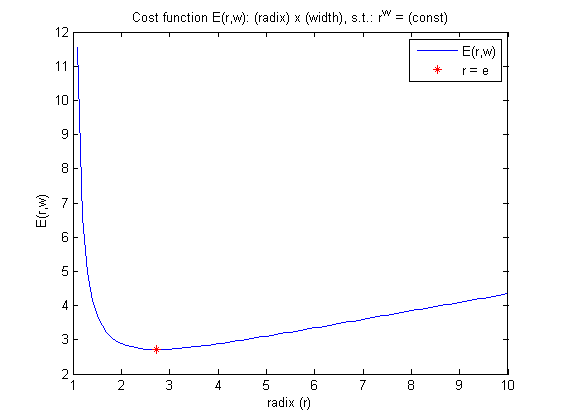}
\end{figure}

\par\end{center}

The question now becomes: \emph{How generic is this assertion?} In
other words, can we use some other optimality criterion and derive
some other solution for the ``best'' radix?

In order to investigate this, we can try and formulate an alternative
cost function and work in a similar way. Although the tree structure
metaphor leads naturally to a product between radix $r$ and width
$w$, the sum of these two parameters may be used as an alternative
test case. Theorem \ref{thm:E2-base-n} exploits this possibility
with $E_{2}\left(r,w\right)$ and states the corresponding solution(s):
\begin{thm}
\label{thm:E2-base-n}Under the minimization criterion $E_{2}\left(r,w\right)=r+w$
and subject to $r^{w}=C\neq0$ (constant), where $r>1$ is the radix
and $w\ \epsilon\left[0,\ r-1\right]$, the optimal information packing
for number representation is achieved for an arithmetic system with
$r$ equal to the root(s) of: $f\left(r\right)=r\cdot\ln^{2}r-\ln C=0$
($C$-dependent).\end{thm}
\begin{proof}
The goal is to minimize: $E_{2}\left(r,w\right)=r+w$, s.t. $r^{w}=C\neq0$
(constant). Again, since $r>1$, from the constraint we have Eq.\ref{eq:w-param-log}.
Substituting in the minimization criterion, then taking the first
derivative \cite{Spiegel2002} and calculating for root(s), gives:

\begin{eqnarray}
\frac{\partial\left(E_{2}\left(r,w\right)\right)}{\partial r}=0 & \Rightarrow & \frac{\partial}{\partial r}\left(r+\frac{\log_{a}C}{\log_{a}r}\right)=1+\log_{a}C\cdot\frac{\partial}{\partial r}\left(\frac{1}{\log_{a}r}\right)=0\nonumber \\
 & \Rightarrow & 1+\ln C\cdot\frac{\partial\left(\ln^{-1}r\right)}{\partial r}=1-\frac{\ln C}{\ln^{2}r}\cdot\frac{\partial\ln r}{\partial r}=0\label{eq:E1-analytical0-1}\\
 & \Rightarrow & 1-\frac{\ln C}{r\cdot\ln^{2}r}=0\Rightarrow r\cdot\ln^{2}r=\ln C\label{eq:E1-analytical1-1}
\end{eqnarray}

\end{proof}
From Eq.\ref{eq:E1-analytical0-1} it becomes clear that the constant
term $\ln C$ will not be removed from the final calculation, which
is confirmed in Eq.\ref{eq:E1-analytical1-1}. As before, it is worth
noting that in Eq.\ref{eq:E1-analytical1-1} there is no reference
to $a$, i.e., the base of the logarithm used in Eq.\ref{eq:w-param-log}
is indeed irrelevant. 

Figure \ref{fig:cost-function-plot2} illustrates the plot of the
first derivative $\frac{\partial E_{2}\left(r,w\right)}{\partial r}$
of the cost function against radix $r$. From Eq.\ref{eq:E1-analytical1-1}
it is evident that the curve is $C$-dependent, therefore the minimization
solutions are too. The plot shows the corresponding curves for various
values of $C$ and Table \ref{tab:E2-solutions} shows the exact solutions
of the associated minimization. Again, the solution is independent
to the base of the logarithm used in the proof, but \emph{dependent}
to the actual upper bound $C$. Therefore, this choice of cost function
does not lead to a generic solution, although the partial solutions
remain asymptotically close to four.

\begin{center}
\begin{figure}
\centering{}\caption{\label{fig:cost-function-plot2}Plot of the $\frac{\partial E_{2}\left(r,w\right)}{\partial r}$
cost function, which is actually $C$-dependent.}
\includegraphics[scale=0.75]{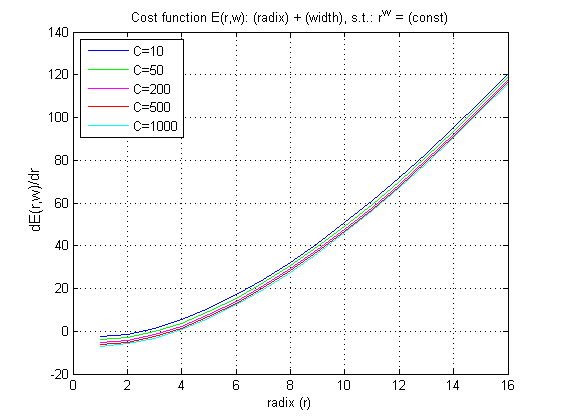}
\end{figure}

\par\end{center}

\begin{center}
\begin{table}

\centering{}\caption{\label{tab:E2-solutions}Sample solutions to $\frac{\partial E_{2}\left(r,w\right)}{\partial r}=0$
for various values of $C$.}
\begin{tabular}{|c|c|}
\hline 
$C$ & $r$\tabularnewline
\hline 
\hline 
$e$ & 1.4215\tabularnewline
\hline 
10 & 2.5746\tabularnewline
\hline 
50 & 3.0841\tabularnewline
\hline 
200 & 3.4519\tabularnewline
\hline 
500 & 3.6724\tabularnewline
\hline 
1000 & 3.8303\tabularnewline
\hline 
\end{tabular}
\end{table}

\par\end{center}

\section{Discussion}

As stated earlier, the formalization of the cost function $E_{1}\left(r,w\right)=r\cdot w$
in Theorem \ref{thm:E1-base-n} is compatible with the tree organization
metaphor of number representation in any arithmetic system. It is
also the actual cost estimator that was employed in the early years
of computing back in the '50s and '60s to predict the hardware cost
of implementing such electronic processing and memory systems \cite{ERA1950}.
The ternary arithmetic, not the binary, seemed to be the right choice
for the computers of the next decades.

Indeed, the first ternary-arithmetic computer\footnote{Actually, one early ternary machine was built entirely from wood by
Thomas Fowler in 1840 and it operated in balanced ternary \cite{TernaryCompWikip2016}.} was built at Moscow State University during the race of the Cold
War by Nikolai P. Brusentsov and his colleagues \cite{Hayes2001,TernaryCompWikip2016}.
It was named ``Setun'' and 50 such machines were built between 1958
and 1965. It operated with arithmetic units of 18 ternary digits or
\emph{trits} (instead of \emph{bits}), producing a numerical range
of $3^{18}=387,420,489$ integer numbers. In contrast, according to
Eq.\ref{eq:w-param-log} a modern binary computer requires $18\cdot\frac{\ln3}{\ln2}=28.529...<29$
bits to represent the same range. However, there were no three-state
electronics circuitry at the time and the machines were built using
two pairs of magnetic cores, i.e., inherently a four-state device
to implement ternary arithmetic. Obviously, this approach produces
hardware that is 25\% less efficient than pure binary and, as a result,
the Setun project was far from success.

Nevertheless, the ternary approach seemed to have many advantages,
not only in relation to information packing in memory circuitry but
also in terms of robust flow control in programs (e.g. three-way ordering/logic
comparisons in a single step). In the '60s several test projects were
developed for ternary logic gates and memory cells for building more
and more complex digital units, such as adders and multiplexers. In
1973 Gideon Frieder and his colleagues at the State University of
New York designed and implemented (in Fortran) a software emulator
for ``Ternac'' \cite{TernacWikip2016}, a ternary computer similar
to Setun.

Donald E. Knuth, one of the founders of modern programming and Informatics
in general, has wrote in his seminal book ``The Art of Computer Programming''
(1968) \cite{Knuth1968,Knuth1997} that ternary arithmetic is \emph{``...perhaps
the prettiest number system of all...''}, explaining the numerous
advantages over binary or any other system. Despite the theoretical
arguments, the ternary approach never made it to mainstream computer
manufacturing, mostly due to the advent of cheap solid-state circuitry
of binary logic that revolutionized the digital technology of the
20th century.

In the last two decades or so, quantum computing \cite{QuantumCompWikip2016}
seems to finally get into solid foundations and hardware implementations,
although it is still constrained to only few \emph{qubits}. One of
the main advantages of the quantum approach to computing is the inherent
multi-valued logic that can be implemented directly in single memory
cells and logic gates \cite{Kubucki2008,QuantumArithm2008}. Obviously,
ternary is a serious candidate as the base of quantum arithmetic,
if the implementation cost of each qubit continues to be even loosely
relevant to the arguments presented here.


\bibliographystyle{elsarticle-num}
\addcontentsline{toc}{section}{\refname}\bibliography{refs}

\end{document}